\documentclass[letterpaper, 10 pt, conference]{ieeeconf}

\IEEEoverridecommandlockouts 
\makeatletter

\let\proof\@undefined
\let\endproof\@undefined
\makeatother
\let\labelindent\relax

\usepackage{jhoagg}
\usepackage{mathtools}
\usepackage{amsfonts}
\usepackage{amssymb,latexsym}
\usepackage[psamsfonts]{eucal}
\usepackage{amsthm}
\usepackage{graphicx}
\usepackage[inline]{enumitem}
\usepackage{indentfirst}
\usepackage{setspace}
\usepackage{microtype}
\usepackage{threeparttable}
\usepackage{float}
\usepackage{cleveref}
\usepackage{afterpage}
\usepackage{placeins}
\usepackage{cancel}
\usepackage{cite}
\usepackage{leftidx}
\usepackage{breqn}
\usepackage[export]{adjustbox}
\usepackage{comment}
\usepackage[ruled, linesnumbered, lined]{algorithm2e}
\usepackage[dvipsnames]{xcolor}
\usepackage{xcolor}
\usepackage{soul}
\usepackage{siunitx}
\usepackage{subcaption}
\usepackage{tabularx}
\usepackage{array} 
\usepackage{url}

\usepackage{tikz}
\usepackage{pgf}
\usepackage{pgfplots}
\usepgflibrary{plothandlers}
\usepgfplotslibrary{external} 
\pgfkeys{/pgf/number format/.cd,1000 sep={}}  
\usetikzlibrary{calc,shapes.misc,plotmarks}
\pgfplotsset{compat=1.13}

\let\originalleft\left
\let\originalright\right
\renewcommand{\left}{\mathopen{}\mathclose\bgroup\originalleft}
\renewcommand{\right}{\aftergroup\egroup\originalright}

\newcounter{thm} 
\newtheorem{theorem}[thm]{\indent Theorem}

\newtheorem{assumption}{\indent Assumption}

\newtheorem{proposition}{\indent Proposition}

\newtheorem{lemma}{\indent Lemma}

\newtheorem{corollary}{\indent Corollary}

\newtheorem{definition}{\indent Definition}

\newtheorem{remark}{\indent Remark}

\newtheorem{example}{\indent Example}

\newtheorem{Simulation}{Simulation}

\newtheorem{fact}{\indent Fact}

\newtheorem{conjecture}{\indent Conjecture}

\newtheorem{experiment}{\indent Experiment}

\allowdisplaybreaks

\renewcommand{\theenumi}{{\it (\alph{enumi})}}

\usepackage{cite}
\usepackage{accents}

\newlength\figureheight 
\newlength\figurewidth

\allowdisplaybreaks

\graphicspath{ {Figures/} }
\DeclareGraphicsExtensions{.png}

\DeclareMathAlphabet{\mathcal}{OMS}{cmsy}{m}{n} 

\crefname{equation}{}{}

\usepackage[ruled, linesnumbered, lined]{algorithm2e}
\SetKwInput{KwInit}{Initialize}
\SetKwFunction{SafeSet}{SafeSet}
\SetKwFunction{GetSafe}{GetSafe}
\SetKwFunction{GetUnsafe}{GetUnsafe}

\SetCommentSty{mycommfont}

\begin{document}
	\title{Receding-Horizon Nonlinear Optimal Control With Safety Constraints Using Constrained Approximate Dynamic Programming}	
	\author{Ricardo Gutierrez and Jesse B. Hoagg  
		\thanks{R. Gutierrez and J. B. Hoagg are with the Department of Mechanical and Aerospace Engineering, University of Kentucky, Lexington, KY, USA. (e-mail: Ricardo.Gutierrez@uky.edu, jesse.hoagg@uky.edu).}
		\thanks{R. Gutierrez is supported by the Fulbright-SENACYT Scholarship. This work is also supported in part by the National Science Foundation (1849213) and Air Force Office of Scientific Research (FA9550-20-1-0028).}
	}    
	\maketitle

\begin{abstract}
We present a receding-horizon optimal control for nonlinear continuous-time systems subject to state constraints.
The cost is a quadratic finite-horizon integral.
The key enabling technique is a new constrained approximate dynamic programming (C-ADP) approach for finite-horizon nonlinear optimal control with constraints that are affine in the control.
The C-ADP approach is intuitive because it uses a quadratic approximation of the cost-to-go function at each backward step.
This method yields a sequence of analytic closed-form optimal control functions, which have identical structure and where parameters are obtained from 2 Riccati-like difference equations. 
This C-ADP method is well suited for real-time implementation.
Thus, we use the C-ADP approach in combination with control barrier functions to obtain a continuous-time receding-horizon optimal control that is farsighted in the sense that it optimizes the integral cost subject to state constraints along the entire prediction horizon. 
Lastly, receding-horizon C-ADP control is demonstrated in simulation of a nonholonomic ground robot subject to velocity and no-collision constraints. 
We compare performance with 3 other approaches. 

%
%
\end{abstract}



\section{Introduction}

Autonomous systems are often required to achieve performance goals (e.g., way-point navigation, formation control) while satisfying state constraints (e.g., safety). 
However, performance goals and constraints may compete with one another.
In this case, conflicts can be resolved using barrier functions~\cite{prajna2007framework,wieland2007constructive,ames2016control} or 
model predictive control (MPC) \cite{borrelli2017predictive,bemporad2002model,tondel2003algorithm}.

Control barrier functions (CBFs) address state constraints by generating controls that make a specified safe set forward invariant \cite{wieland2007constructive}. 
CBFs are often integrated in real-time optimization as constraints that ensure forward invariance of the safe set while minimizing an instantaneous cost \cite{ames2016control}. 
One approach is instantaneous minimum intervention, where a CBF safety filter generates a control that is as close as possible to a desired performance control while ensuring state constraint satisfaction. 
CBF methods have been demonstrated on a variety of applications, including  
mobile ground robots \cite{ames2019control,cosner2021measurement, safari2025time,safari2025safe,rabiee2024composition,rabiee2024closed,rabiee2023automatica}, unmanned air vehicles \cite{wang2023multi, zheng2023constrained}, spacecraft \cite{breeden2023robust,kamat2026electromagnetic}, and cruise control \cite{ames2016control,xiao2022event}. 
CBF approaches can be  computationally efficient (e.g., implemented as quadratic programs); however, these methods tend to be shortsighted and greedy, prioritizing 
instantaneous safety and performance~\cite{rabiee2025guaranteed}.

On the other hand, MPC optimizes performance over a forward-looking time horizon. 
In the most general setting, MPC methods can address nonlinear dynamics, nonlinear objective functions, and constraints.
However, nonlinear MPC must be initialized carefully to ensure recursive feasibility of the constraints, and it is computationally expensive.
Thus, methods have been developed to address the computational burden. 
For example, the iterative linear quadratic regulator (iLQR) uses Taylor-series expansion to approximate the nonlinear optimal control problem to a linear-quadratic problem that can be solved with finite-horizon time-varying LQR \cite{li2004iterative}. 
Then, iterative forward-and-backward passes are used to refine the solution; these passes may include line search \cite{iLQR_Bharath2}.
However, convergence may be slow \cite{sleiman2021constraint}.
Constrained iLQR (CiLQR) is an extension to address constraints by including penalization in the cost function and using augmented Lagrangians \cite{aoyama2021constrained,sleiman2021constraint}. 
Although CiLQR encourages constraint satisfaction, constraints can be violated.

Dynamic programming is a useful tool for problems that involve sequential decisions such as the multistage optimizations that arise in finite-horizon optimal control \cite{rao2019engineering,bertsekas2011dynamic}. 
Dynamic programming involves solving the Bellman equation, which is generally difficult with nonquadratic cost, nonlinear dynamics, or constraints \cite{ha2022novel}. 
Thus, approximate dynamic programming (ADP) uses approximations to simplify the recursive solution of the Bellman equation \cite{bertsekas2011dynamic}. 
For example, \cite{ha2022novel,xiao2015online,wang2010adaptive,guo2017policy} use neural networks to approximate the cost-to-go function. 
However, this can be computationally expensive and does not directly address constraints. 
Approaches that address input constraints include \cite{yuan2020solver,yang2019approximate,mare2007solution}. 
For example, \cite{mare2007solution} uses dynamic programming to obtain an analytical solution for an input-constrained LQR problem, where the dynamics are linear and the scalar input is upper and lower bounded. 
Notably, \cite{yuan2020solver,yang2019approximate,mare2007solution} do not address state constraints.

This article presents 2 primary contributions. 
First, Section~\ref{sect:ADP-GENERAL} addresses discrete-time finite-horizon nonlinear optimal control with a quadratic cost, nonlinear dynamics, and nonlinear constraints, which are affine in the control at each step of the horizon. 
We present an analytic closed-form solution using constrained approximate dynamic programming (C-ADP), where the cost-to-go function is quadratically approximated at each step of the backward recursion.
This approach yields a sequence of analytic nonlinear control functions that have identical structure, where the parameters of each function are obtained through the recursive solution of 2 Riccati-like difference equations. 
This analytic solution is well suited for real-time implementation, and the sequence of analytic optimal functions can be updated efficiently.
These features are important for the article's other contribution.

The second contribution is a receding-horizon nonlinear optimal control for continuous-time systems subject to state constraints. 
We use the C-ADP solution to construct a feedback control that approximately minimizes a finite-horizon integral cost subject to state constraints that are enforced along the horizon using affine constraints generated from a CBF. 
Notably, the sequence of C-ADP optimal control functions ensures that state constraints are satisfied not only instantly but for all time in the horizon. 
Note that although the C-ADP approach is discrete, the optimal control functions can be implemented continuously to ensure constraint satisfaction for the continuous-time dynamics. 
Finally, Section~\ref{sec:robot} presents numerical simulations of a nonholonomic ground robot navigating a cluttered environment. 
We compare the performance of receding-horizon C-ADP control with 2 variants of CiLQR and a greedy minimum-intervention control.

\section{Constrained Approximate Dynamic Programming (C-ADP)} \label{sect:ADP-GENERAL}

Consider the quadratic cost $J:\mathbb{R}^{n} \times \cdots \times \mathbb{R}^{n} \times \mathbb{R}^{m} \times \cdots \times \mathbb{R}^{m} \to \mathbb{R}$ defined by
\begin{equation}
    J(x_1,\ldots, x_N,u_0,\ldots,u_{N-1} )\triangleq l_{N}(x_{N})+\sum_{i=0}^{N-1} l_{i}(x_{i},u_{i}),
    \label{ADP:Total_cost}
\end{equation}
where stage cost $l_{i}:\mathbb{R}^{n} \times \mathbb{R}^{m} \to \mathbb{R}$ and terminal cost $l_{N}:\mathbb{R}^{n}\to \mathbb{R}$ are 
\begin{gather}
    l_{i}(x_i,u_i) \triangleq \frac{1}{2}u_i^{\rm{T}}R_{i}u_i+\Omega_{i}^{\rm{T}}u_i+ \frac{1}{2}x_i^{\rm{T}}Q_{i}x_i+\Gamma_{i}^{\rm{T}}x_i,
    \label{ADP:li}\\
    l_{N}(x_{N}) \triangleq \frac{1}{2}x_{N}^{\rm{T}}Q_{N}x_{N}+\Gamma_{N}^{\rm{T}}x_{N},
    \label{ADP:lN}
\end{gather}
where $x_{0}\in\mathbb{R}^{n}$ is given; $R_0,\ldots,R_{N-1} \in \mathbb{R}^{m \times m}$ are positive definite; $Q_{0},\ldots,Q_N\in \mathbb{R}^{n \times n}$ are positive semidefinite; $\Omega_0,\ldots,\Omega_{N-1} \in \mathbb{R}^{m}$; and  $\Gamma_{0},\ldots,\Gamma_{N} \in \mathbb{R}^{n}$.

The objective of this section is to solve the following constrained nonlinear optimization problem 
\begin{subequations}\label{ADP:OCP:Main_prob}
        \begin{align}
            \min_{\substack{u_{0},\ldots,u_{N-1} \\
                  x_{1},\ldots,x_{N}}}  &J(x_1,\ldots, x_N,u_0,\ldots,u_{N-1} )              \label{ADP:OCP:Main_prob:cost}\\
            \text{subject to} \qquad &x_{i+1} = F(x_i)+G(x_i)u_{i}, \quad i\in\SN,  \label{ADP:OCP:Main_prob:eq_const}\\
             &a_{i}(x_{i})+b_{i}(x_{i})^{\rm{T}}u_{i}   \ge 0, \quad i\in\SN,                   \label{ADP:OCP:Main_prob:ineq_const}
        \end{align}
    \end{subequations} 
where $F:\mathbb{R}^{n} \to \mathbb{R}^{n}$, $G:\mathbb{R}^{n} \to \mathbb{R}^{n \times m}$, $a_1,\ldots,a_{N-1}:\mathbb{R}^{n} \to \mathbb{R}$, and $b_1,\ldots,b_{N-1}:\mathbb{R}^{n} \to \mathbb{R}^{m}$ are continuously differentiable; and $\SN \triangleq \{ 0, \ldots, N-1\}$.

We restrict our attention to quadratic cost $J$; however, non-quadratic costs can be considered using 2nd-order Taylor-series approximation if appropriate. 
We make the following assumption to ensure feasibility.  

\begin{assumption}\rm
    For all $i\in \SN$ and all $x \in \mathbb{R}^{n}$, if $b_{i}(x)=0$, then $a_i(x)>0$.
         \label{assum:b(x)_neq_0}
\end{assumption}

The method in this article uses approximate dynamic programming to obtain a closed-form solution (i.e., $u_i$ as a function of $x_i$) that satisfies constraints \cref{ADP:OCP:Main_prob:eq_const,ADP:OCP:Main_prob:ineq_const} but may be suboptimal because a 2nd-order expansion of the cost-to-go function along a nominal trajectory $\{\bar{x}_{i}\}_{i\in\SN}$ is used in each step of the dynamic programming recursion. 
The next subsection demonstrates the solution approach for an $N=3$ horizon. 
Then, we present the $N$-step horizon closed-form C-APD solution and analyze its properties.

\subsection{Solution Derivation for $N=3$}
\label{sec:derivation}
    
Consider the optimization problem \eqref{ADP:OCP:Main_prob}, where $N=3$. 
The optimizer for step $i=2$ is $u_{2}^{*}:\mathbb{R}^{n} \to \mathbb{R}^{m}$ given by 
\begin{subequations}\label{ADP:OCP:Main_prob3}
    \begin{align}
    u_{2}^{*}(x_{2})  & \triangleq \arg\min_{u_{2}} l_{2}\big(x_{2},u_{2}\big)+l_{3}(x_{3})     \label{ADP:OCP:approx:cost3}\\
    \text{subject to}\quad   & x_{3} = F(x_{2})+G(x_{2})u_{2},   \label{ADP:OCP:approx:eq_const3}\\
     & a_{2}(x_{2})+b_{2}(x_{2})^{\rm{T}}u_{2}   \ge 0.    \label{ADP:OCP:approx:ineq_const3}
    \end{align}
\end{subequations} 
Since \eqref{ADP:OCP:Main_prob3} is a QP, the first-order necessary conditions yield the minimizer
\begin{equation}
u_{2}^{*}(x_{2}) =  k_{2}(x_{2})+\lambda_{2}(x_{2})W_{2}(x_{2}) b_{2}(x_{2}), \label{ADP:OCP:approx:u_2}
\end{equation}        
where 
\begin{align}
\lambda_{2}(x_{2}) &\triangleq \max \left \{ 0, \frac{-a_{2}(x_{2})-b_{2}(x_{2})^{\rm{T}}k_{2}(x_{2})}{b_{2}(x_{2})^{\rm{T}}W_{2}(x_{2})b_{2}(x_{2})} \right \},\label{ADP:OCP:approx:lamb_2}\\
W_{2}(x_{2}) &\triangleq (R_{2}+G(x_{2})^{\rm{T}}P_{3}G(x_{2}))^{-1}, \label{ADP:OCP:approx:W_2}\\
k_{2}(x_{2}) &\triangleq -W_2(x_2) \left [ G(x_{2})^{\rm{T}} (P_{3}F(x_{2}) +T_{3})+\Omega_{2} \right ], \label{ADP:OCP:approx:c_2}
\end{align}
where $P_{3}=Q_{3}$ and $T_{3}=\Gamma_{3}$.
The optimal cost is 
\begin{equation} 
l_{2}^{*}(x_{2}) \triangleq l_{2}\big(x_{2},u_{2}^{*}(x_{2})\big)+l_{3} ( F_2^*(x_2) ), \label{ADP:OCP:approx:cost2go_2}
\end{equation}
where $F_2^*(x_2) \triangleq F(x_{2}) + G(x_2)u_{2}^{*}(x_{2})$.

The closed-form optimizer \eqref{ADP:OCP:approx:u_2} for this first step is the unique global minimizer of \eqref{ADP:OCP:Main_prob3}.
The next step in dynamic programming is to find the optimal $u_1$ by solving 
\begin{subequations}\label{eq:opt.u_1}
        \begin{align}
            \min_{u_{1}} \, \, &l_{1}\big(x_{1},u_{1}\big)+ l_{2}^*(x_{2}) \label{eq:C2G.u_1}\\
            \textrm{subject to} \qquad &  x_{2}=F(x_{1})+G(x_{1})u_{1},  \\
            & a_{1}(x_{1})+b_{1}(x_{1})^{\rm{T}}u_{1} \geq 0.  
        \end{align}
\end{subequations}
However, \eqref{eq:opt.u_1} cannot generally be solved in closed form. 
Thus, we solve a surrogate of this problem by approximating the cost-to-go $l_2^*$ in \eqref{eq:C2G.u_1}. 
First, \eqref{ADP:OCP:approx:lamb_2} is approximated by replacing the maximum with the softplus. 
Specifically, let $\eta>0$, and consider $\mathrm{softplus}_{\eta}:\mathbb{R} \to [0,\infty)$ defined by
\begin{equation}
    \mathrm{softplus}_{\eta}(z)= \frac{1}{\eta}\log\Big(1+\mathrm{e}^{\eta z}\Big),
\end{equation}
which is a smooth upper-bound approximation of $\max \{ 0,z\}$.
Thus, $\lambda_2$ is approximated as 
\begin{equation}
\tilde{\lambda}_{2}(x_{2}) \triangleq \mbox{softplus}_\eta \left ( \frac{-a_{2}(x_{2})-b_{2}(x_{2})^{\rm{T}}k_{2}(x_{2})}{b_{2}(x_{2})^{\rm{T}}W_{2}(x_{2})b_{2}(x_{2})} \right ),\label{ADP:OCP:approx:tLambda_2}
\end{equation}
and the associated approximations of $u_2^*$ and $F_2^*$ are
\begin{align}
        \tilde{u}_{2}^{*}(x_{2})  &\triangleq k_{2}(x_{2})+\tilde{\lambda}_{2}(x_{2})W_2(x_2)b_{2}(x_{2}) ,
        \label{ADP:OCP:approx:tu_2}\\
         \tilde{F}_{2}^{*}(x_{2}) &\triangleq F(x_{2})+G(x_{2})\tilde{u}_{2}^{*}(x_{2}). 
        \label{ADP:OCP:approx:tx_2}
\end{align}  
Then, the cost-to-go \eqref{ADP:OCP:approx:cost2go_2} can be approximated by
\begin{equation}\label{tilde_l2}
\tilde{l}_{2}^*(x_{2})\triangleq l_{2}\big(x_{2},\tilde{\tilde{u}}_{2}^*(x_{2})\big)+l_{3}(\tilde{\tilde{F}}_{2}^*(x_{2})),
\end{equation}    
where
\begin{align}
         \tilde{\tilde{u}}_{2}^*(x_{2}) &\triangleq K_{2}(x_{2}-\bar x_{2})+ u_{2}^{*}(\bar{x}_{2}), \label{eq:tilde_tilde_u2}\\
         \tilde{\tilde{F}}_{2}^*(x_{2}) &\triangleq \tilde{A}_{2}(x_{2}-\bar x_{2})+ F_{2}^{*}(\bar{x}_{2})\label{eq:tilde_tilde_F2}
\end{align} 
are the Taylor-series approximation of $u_2^*$ and $F_2^*$, and 
\begin{equation}\label{eq:K2}
K_{2} \triangleq \frac{\partial \tilde{u}_{2}^{*}(x)}{\partial x} \Bigg|_{\bar x_{2}},\qquad
\tilde{A}_{2} \triangleq \frac{\partial \tilde F_2^*(x)}{\partial x} \Bigg|_{\bar x_{2}}.
\end{equation}
Then, using \cref{ADP:li,ADP:lN,eq:tilde_tilde_u2,eq:tilde_tilde_F2,eq:K2} in \eqref{tilde_l2} yields
\begin{equation}
\tilde{l}_{2}^*(x_{2}) = \frac{1}{2}x_{2}^{\rm{T}}P_{2}x_{2}+T_{2}^{\rm{T}}x_{2}, \label{ADP:OCP:approx:approx_cost2go_2}
\end{equation}
where     
\begin{align}
P_{2} &\triangleq Q_{2}+K_{2}^{\rm{T}}R_{2}K_{2}+\tilde{A}_{2}^{\rm{T}}P_{3}\tilde{A}_{2}, \label{ADP:OCP:approx:P_2}\\
T_{2} &\triangleq K_{2}^{\rm{T}}\Big(R_{2} u_{2}^{*}(\bar{x}_{2})-R_{2}K_{2}\bar{x}_{2}+\Omega_{2}\Big)+\Gamma_{2}\nn\\
& \qquad +\tilde{A}_{2}^{\rm{T}}\Big(P_{3} F_{2}^{*}(\bar{x}_{2})-P_{3}\tilde{A}_{2}\bar{x}_{2}+T_{3}\Big).    \label{ADP:OCP:approx:T_2}
\end{align}

Then, we solve \eqref{eq:opt.u_1} with $l_2^*$ is replaced by $\tilde l_2^*$. 
Specifically, the optimizer for step $i=1$ is $u_1^*:\mathbb{R}^{n} \to \mathbb{R}^{m}$ given by
    \begin{subequations}         \label{ADP:OCP:approx:QP_form_1}
        \begin{align}
            u_{1}^{*}(x_{1}) & \triangleq \arg\min_{u_{1}} l_{1}\big(x_{1},u_{1}\big)+ \tilde{l}_{2}^*(x_{2}) \label{ADP:OCP:approx:cost1}\\
            \textrm{subject to} \quad &  x_{2}=F(x_{1})+G(x_{1})u_{1},  \label{ADP:OCP:approx:eq_const1}\\
            & a_{1}(x_{1})+b_{1}(x_{1})^{\rm{T}}u_{1} \geq 0.  \label{ADP:OCP:approx:ineq_const1}
        \end{align}  
    \end{subequations}
Note that \eqref{ADP:OCP:approx:QP_form_1} is a QP with structure that mirrors \eqref{ADP:OCP:Main_prob3} from the previous backward step. 
The first-order necessary conditions yield the minimizer
\begin{equation}
u_{1}^{*}(x_{1}) =  k_{1}(x_{1})+\lambda_{1}(x_{1})W_{1}(x_{1}) b_{1}(x_{1}), \label{ADP:OCP:approx:u_1}
\end{equation}        
where 
\begin{align}
\lambda_{1}(x_{1}) &\triangleq \max \left \{ 0, \frac{-a_{1}(x_{1})-b_{1}(x_{1})^{\rm{T}}k_{1}(x_{1})}{b_{1}(x_{1})^{\rm{T}}W_{1}(x_{1})b_{1}(x_{1})} \right \},\label{ADP:OCP:approx:lamb_1}\\
W_{1}(x_{1}) &\triangleq (R_{1}+G(x_{1})^{\rm{T}}P_{2}G(x_{1}))^{-1}, \label{ADP:OCP:approx:W_1}\\
k_{1}(x_{1}) &\triangleq -W_1(x_1) \left [ G(x_{1})^{\rm{T}} (P_{2}F(x_{1}) +T_{2})+\Omega_{1} \right ], \label{ADP:OCP:approx:c_1}
\end{align}
where $P_{3}$ and $T_{3}$ are given by \cref{ADP:OCP:approx:P_2,ADP:OCP:approx:T_2}.
Notice that the optimizer \cref{ADP:OCP:approx:u_1,ADP:OCP:approx:lamb_1,ADP:OCP:approx:W_1,ADP:OCP:approx:c_1} has the same form as the optimizer \cref{ADP:OCP:approx:u_2,ADP:OCP:approx:lamb_2,ADP:OCP:approx:W_2,ADP:OCP:approx:c_2} on the previous step with subscripts 2 and 3 replaced by 1 and 2, respectively. 
Similarly, the optimal cost is 
\begin{equation} 
l_{1}^{*}(x_{1}) \triangleq l_{1}\big(x_{1},u_{1}^{*}(x_{1})\big)+ \tilde l_{2}^* ( F_1^*(x_1) ), \label{ADP:OCP:approx:cost2go_1}
\end{equation}
where $F_1^*(x_1) \triangleq F(x_1) + G(x_1)u_1^{*}(x_1)$.

The next backward step is to find the optimal $u_0$ by solving the optimization \eqref{eq:opt.u_1} with subscripts 1 and 2 replaced by 0 and 1, respectively. 
To solve this optimization, we use the same steps as before. 
We replace the cost-to-go $l_1^*$ by the approximation 
\begin{equation}\label{eq:tilde_l1}
\tilde{l}_{1}^*(x_{1})\triangleq l_{1}\big(x_{1},\tilde{\tilde{u}}_{1}^*(x_{1})\big)+\tilde l_{2}^*(\tilde{\tilde{F}}_{1}^*(x_{1})),
\end{equation}   
where $\tilde{\tilde{u}}_{1}^*$ and $\tilde{\tilde{F}}_{1}^*(x_{1})$ are given by \cref{eq:tilde_tilde_u2,eq:tilde_tilde_F2,eq:K2} and \cref{ADP:OCP:approx:tLambda_2,ADP:OCP:approx:tu_2,ADP:OCP:approx:tx_2} with subscript 2 replaced by subscript 1. 
Then, substituting \cref{ADP:li,ADP:li,ADP:OCP:approx:approx_cost2go_2} into \eqref{eq:tilde_l1} yields
\begin{equation}
\tilde{l}_{1}^*(x_{1}) = \frac{1}{2}x_{1}^{\rm{T}}P_{1}x_{1}+T_{1}^{\rm{T}}x_{1}, \label{ADP:OCP:approx:approx_cost2go_1}
\end{equation}
where 
\begin{align}
P_{1} &\triangleq Q_{1}+K_{1}^{\rm{T}}R_{1}K_{1}+\tilde{A}_{1}^{\rm{T}}P_{2}\tilde{A}_{1}, \label{ADP:OCP:approx:P_1}\\
T_{1} &\triangleq K_{1}^{\rm{T}}\Big(R_{1} {u}_{1}^{*}(\bar{x}_{1})-R_{1}K_{1}\bar{x}_{1}+\Omega_{1}\Big)+\Gamma_{1}\nn\\
& \qquad +\tilde{A}_{1}^{\rm{T}}\Big(P_{2} {F}_{1}^{*}(\bar{x}_{1})-P_{2}\tilde{A}_{1}\bar{x}_{1}+T_{2}\Big). \label{ADP:OCP:approx:T_1}
\end{align}
The approximate cost-to-go \cref{ADP:OCP:approx:approx_cost2go_1,ADP:OCP:approx:P_1,ADP:OCP:approx:T_1} has the same form as the approximate cost-to-go \cref{ADP:OCP:approx:approx_cost2go_2,ADP:OCP:approx:P_2,ADP:OCP:approx:T_2} on the previous step with subscripts 2 and 3 replaced by 1 and 2, respectively.
The optimizer for step $i=0$ is $u_0^*:\mathbb{R}^{n} \to \mathbb{R}^{m}$ given by
    \begin{subequations} \label{ADP:OCP:approx:QP_form_0}
        \begin{align}
            u_{0}^{*}(x_{0}) & \triangleq \arg\min_{u_{0}} l_{0}\big(x_{0},u_{0}\big)+ \tilde{l}_{1}^*(x_{1})\\
            \textrm{subject to} \quad &  x_{1}=F(x_0)+G(x_0)u_0, \\
            & a_{0}(x_0)+b_{0}(x_0)^{\rm{T}}u_0 \geq 0.  
        \end{align}  
    \end{subequations}
The QP \eqref{ADP:OCP:approx:QP_form_0} has the same structure as
\cref{ADP:OCP:approx:QP_form_1,ADP:OCP:Main_prob3} from the previous steps, and the unique global minimizer has the same form as the previous steps. 
Specifically, the optimizer is given by \cref{ADP:OCP:approx:u_1,ADP:OCP:approx:lamb_1,ADP:OCP:approx:W_1,ADP:OCP:approx:c_1} with subscripts 1 and 2 replaced by 0 and 1, respectively.

\subsection{Solution for $N$-Step Horizon}

The derivation process from the previous section can be applied for an $N$-step horizon. 
This process yields the following solution.
For all $i \in \SN$, define 
\begin{equation}
u_{i}^{*}(x)\triangleq k_{i}(x) + \lambda_{i}(x)W_{i}(x) b_{i}(x), \label{Gral_alg:ctrl_optim}
\end{equation}
where
\begin{align}
        \lambda_{i}(x) &\triangleq \max \left \{ 0, \frac{-a_{i}(x)-b_{i}(x)^{\rm{T}}k_{i}(x)}{b_{i}(x)^{\rm{T}}W_{i}(x)b_{i}(x)} \right \},
        \label{Gral_alg:Lambda_optim}\\
        W_{i}(x) &\triangleq \left (R_{i}+G(x)^{\rm{T}}P_{i+1}G(x) \right )^{-1},
        \label{Gral_alg:W_optim}\\
        k_{i}(x) &\triangleq  -W_i(x) \left [ G(x)^{\rm{T}} \left ( P_{i+1}F(x)+T_{i+1} \right )+\Omega_{i} \right ],
        \label{Gral_alg:c_optim}
\end{align}
where
\begin{align}
P_{i}&\triangleq Q_{i}+K_{i}^{\rm{T}}R_{i}K_{i}+\tilde{A}_{i}^{\rm{T}}P_{i+1}\tilde{A}_{i},
        \label{Gral_alg:P_optim}\\
T_{i} &\triangleq \tilde{A}_{i}^{\rm{T}} \left ( P_{i+1} ( F(\bar x_i) + G(\bar x_i){u}_{i}^{*}(\bar x_i) - \tilde{A}_{i}\bar{x}_{i} ) +T_{i+1} \right )\nn\\
& \qquad + K_{i}^{\rm{T}}\left (R_{i} {u}_{i}^{*}(\bar{x}_{i})-R_{i}K_{i}\bar{x}_{i}+\Omega_{i} \right )+\Gamma_{i}
        \label{Gral_alg:T_optim}
\end{align}
where $P_{N}=Q_{N}$, $T_{N}=\Gamma_{N}$, and 
\begin{gather}
        \tilde{A}_{i} \triangleq \frac{\partial \tilde F_i^*(x)}{\partial x} \Bigg|_{\bar x_{i}},  \qquad
        K_{i} \triangleq \frac{\partial \tilde{u}_{i}^{*}(x)}{\partial x} \Bigg|_{\bar x_{i}},
        \label{Gral_alg:AK_optim}\\
        \tilde{\lambda}_{i}(x) \triangleq \mbox{softplus}_\eta \left ( \frac{-a_{i}(x)-b_{i}(x)^{\rm{T}}k_{i}(x)}{b_{i}(x)^{\rm{T}}W_{i}(x)b_{i}(x)} \right ),
        \label{Gral_alg:lambda_tild_optim}\\
\tilde u_{i}^{*}(x)\triangleq k_{i}(x) +\tilde \lambda_{i}(x)W_{i}(x) b_{i}(x) 
        \label{Gral_alg:u_tild_optim},\\
    \tilde{F}_{i}^*(x)\triangleq F(x)+G(x)\tilde{u}_{i}^*(x).
        \label{Gral_alg:Fxu_tild_optim}
\end{gather}    
    
The sequence of optimal functions $u_{N-1}^*,\ldots,u_0^*$ are solved using backward recursion. Specifically, \cref{Gral_alg:ctrl_optim,Gral_alg:Lambda_optim,Gral_alg:W_optim,Gral_alg:c_optim,Gral_alg:P_optim,Gral_alg:T_optim} are solved from $i=N-1$ to $i=0$, where the the recursion is initialized with the terminal conditions $P_N=Q_N$ and $T_N=\Gamma_N$. 
Since $Q_i$ is positive semidefinite and $R_i$ is positive definite, it follows from \cref{Gral_alg:W_optim,Gral_alg:P_optim} that $P_{i}$ is positive semidefinite and for all $x \in \BBR^n$, $W_{i}(x)$ exists and is positive definite.

\begin{remark}\rm
The C-ADP solution 
\cref{Gral_alg:ctrl_optim,Gral_alg:Lambda_optim,Gral_alg:W_optim,Gral_alg:c_optim,Gral_alg:P_optim,Gral_alg:T_optim,Gral_alg:AK_optim,Gral_alg:lambda_tild_optim,Gral_alg:u_tild_optim,Gral_alg:Fxu_tild_optim} simplifies to the classic finite-horizon linear quadratic regulator (LQR) for discrete time. 
To illustrate this case, for all $i\in\SN$, let $R_i =R \in \BBR^{m \times m}$, $Q_i = Q \in \BBR^{n \times n}$, $\Omega_i=0$, and $\Gamma_i=0$, which implies that \eqref{ADP:Total_cost} is the finite-horizon LQR cost.
Next, let $F(x) = Ax$ and $G(x)=B$, where $A\in\BBR^{n \times n}$ and $B\in\BBR^{n\times m}$, which implies that \eqref{ADP:OCP:Main_prob:eq_const} is linear (i.e., discrete-time linear dynamics).
Finally, let $a_i(x) > 0$ and $b_i=0$, which implies that there are no inequality constraints (i.e., \eqref{ADP:OCP:Main_prob:ineq_const} are trivially satisfied). 
In this case, it follows from \cref{Gral_alg:ctrl_optim,Gral_alg:Lambda_optim,Gral_alg:W_optim,Gral_alg:c_optim,Gral_alg:P_optim,Gral_alg:T_optim,Gral_alg:AK_optim,Gral_alg:lambda_tild_optim,Gral_alg:u_tild_optim,Gral_alg:Fxu_tild_optim} that for $i\in \SN$, $T_i=0$, $\lambda_i=0$, and $W_i = (R+B^\rmT P_{i+1} B)^{-1}$.
Furthermore, for $i\in \SN$, the optimizer is 
\begin{equation*}
u_i^*(x) = K_i x, \quad K_i = -(R+B^\rmT P_{i+1} B)^{-1} B^\rmT P_{i+1} A,
\end{equation*}
where \eqref{Gral_alg:P_optim} simplifies to the discrete Riccati equation
\begin{equation*}
P_i = Q + A^\rmT P_{i+1} A - A^\rmT P_{i+1} B (R + B^\rmT P_{i+1} B)^{-1} B^\rmT P_{i+1} A,
\end{equation*}
where $P_N=Q_N$. 
In this case, $K_i$ is the LQR gain.
\end{remark}

The following result shows that  $u_i^*$ satisfies \cref{ADP:OCP:Main_prob:ineq_const} for all $x$.

\begin{theorem} \label{th:ADP:const>=0}\rm
Assume Assumption~\ref{assum:b(x)_neq_0} is satisfied.
Then, for all $i\in \SN$ and all $x\in\BBR^n$, $a_i(x)+b_i(x)^{\rm{T}}u_{i}^{*}(x) \geq 0$.
\end{theorem}

\begin{proof}[\indent Proof]
Since $W_i(x)$ is positive definite, it follows from \cref{Gral_alg:ctrl_optim,Gral_alg:Lambda_optim} that
\begin{align*}
a_{i}(x)+b_{i}(x)^{\rm{T}}u_{i}^{*}(x)&= a_{i}(x)+b_{i}(x)^{\rm{T}}k_{i}(x)\\
        &\quad +\max \, \{0,-a_{i}(x)-b_{i}(x)^{\rm{T}}k_{i}(x) \},  \\  
        &= \max \, \{0,a_{i}(x)+b_{i}(x)^{\rm{T}}k_{i}(x)\},
    \end{align*}
which is nonnegative.
\end{proof}

The optimal functions $u_{i}^*$ have identical structure for each step $i$, and the information required to construct $u_{i}^*$ is propagated through the backward recursion of the matrices $P_i$ and $T_i$. 
Thus, the C-ADP solution 
\cref{Gral_alg:ctrl_optim,Gral_alg:Lambda_optim,Gral_alg:W_optim,Gral_alg:c_optim,Gral_alg:P_optim,Gral_alg:T_optim,Gral_alg:AK_optim,Gral_alg:lambda_tild_optim,Gral_alg:u_tild_optim,Gral_alg:Fxu_tild_optim} is well suited for a real-time receding-horizon implementation. 
Unlike iLQR or CilQR, the solution \cref{Gral_alg:ctrl_optim,Gral_alg:Lambda_optim,Gral_alg:W_optim,Gral_alg:c_optim,Gral_alg:P_optim,Gral_alg:T_optim,Gral_alg:AK_optim,Gral_alg:lambda_tild_optim,Gral_alg:u_tild_optim,Gral_alg:Fxu_tild_optim} satisfies the constraint \cref{ADP:OCP:Main_prob:ineq_const}, and the solution is in closed form, which facilitates fast-and-efficient feedback implemented as presented in Section~\ref{sect:control+CBF}.

\subsection{Optimality of the Solution}

Although \cref{Gral_alg:ctrl_optim,Gral_alg:Lambda_optim,Gral_alg:W_optim,Gral_alg:c_optim,Gral_alg:P_optim,Gral_alg:T_optim,Gral_alg:AK_optim,Gral_alg:lambda_tild_optim,Gral_alg:u_tild_optim,Gral_alg:Fxu_tild_optim} is not necessarily optimal with respect to the original cost \cref{ADP:OCP:Main_prob:cost}, $u_i^*$ is optimal with respect to the approximate stage cost obtained by replacing the cost-to-go with the approximate cost-to-go. 
%
%
%
%
To more clearly explain, we define the approximate cost-to-go for the $i-1$ step of recursion. 
Specifically, for all $i\in \SN$, define 
\begin{align}
\mspace{-12mu} \tilde{l}_{i}^*&(x_i) \triangleq l_{i}\Big(x_i,K_i(x_i-\bar x_i)+ u_i^{*}(\bar{x}_i)\Big)\nn\\
&\quad+\tilde l_{i+1}^*\Big(\tilde{A}_i(x_i-\bar x_i)+ F(\bar x_i) + G(\bar x_i)u_i^{*}(\bar x_i) \Big),\label{eq:approx_ctg}
\end{align}   
where $\tilde l^*_{N}(x_N) = l_N(x_N)$. 
The approximate cost-to-go \eqref{eq:approx_ctg} is the $N$-step extension of \cref{tilde_l2,eq:tilde_l1} used in the $N=3$ derivation.
Similar to Section~\ref{sec:derivation}, using \cref{ADP:li} in \eqref{eq:approx_ctg} yields 
\begin{equation}
\tilde{l}_{i}^*(x_i) = \frac{1}{2} x_i^\rmT P_i x_i + T_i^\rmT x_i.\label{eq:approx_ctg.2}
\end{equation}   
Then, for all $i\in \SN$, the approximate cost at step $i$ is 
\begin{equation}
\tilde{J}_{i}(x_i,u_i) \triangleq l_i(x_i,u_i) + \tilde{l}_{i+1}^*\Big (F(x_i)+G(x_i)u_i \Big ), \label{eq:approx_stage}
\end{equation}
and using \cref{ADP:li,eq:approx_ctg.2} in \eqref{eq:approx_stage} yields
\begin{align}
\tilde{J}_{i}(x_i,u_i) 
&= \frac{1}{2} u_i^\rmT \Big[ R_i + G(x_i)^\rmT P_{i+1} G(x_i) \Big ] u_i\nn\\
&\qquad+ \Big [ G(x_i)^\rmT \left ( T_{i+1} + P_{i+1} F(x_i) \right ) + \Omega_{i} \Big]^{\rm{T}} u_i \nn\\
&\qquad + \frac{1}{2} x_i^{\rm{T}}Q_{i} x_i+ \frac{1}{2} F(x_i)^\rmT P_{i+1} F(x_i)\nn\\
&\qquad+\Gamma_{i}^{\rm{T}} x_i + T_{i+1}^\rmT F(x_i), \label{eq:approx_stage.2}
\end{align}
which is quadratic and strictly convex in $u_i$. 
The next result shows that for all $x \in \BBR^n$, $u = u^*_i(x)$ is the unique global minimizer of $\tilde{J}_{i}(x,u)$ subject to the constraint $a_i(x) +b_i(x)^\rmT u \ge 0$.
The proof is similar to that of \cite[Theorem 1]{safari2025time} and is omitted due to space constraints. 

\begin{theorem}\label{thm:optimal}
\rm
Let $i\in\SN$ and $x \in \BBR^n$, and assume Assumption \ref{assum:b(x)_neq_0} is satisfied.
Let $u\in\BBR^m$ be such that $u \ne u_i^*(x)$ and $a_i(x) +b_i(x)^\rmT u \ge 0$. 
Then, $\tilde{J}_{i}(x,u) > \tilde{J}_{i}(x,u_i^*(x))$.
\end{theorem}

\section{Closed-Form Receding-Horizon Nonlinear Optimal Control with State Constraints}
\label{sect:control+CBF}

This section presents a receding-horizon optimal control using the closed-form C-ADP solution \cref{Gral_alg:ctrl_optim,Gral_alg:Lambda_optim,Gral_alg:W_optim,Gral_alg:c_optim,Gral_alg:P_optim,Gral_alg:T_optim,Gral_alg:AK_optim,Gral_alg:lambda_tild_optim,Gral_alg:u_tild_optim,Gral_alg:Fxu_tild_optim}. 
Consider the dynamic system
    \begin{equation}
		\Dot{x}(t)=f(x(t))+g(x(t))v(t),
		\label{eq:dyn}
    \end{equation}
    where $x(t) \in \mathbb{R}^n$ is the state; $x(0)=x_{0} \in \BBR^n$ is the initial condition; $v: [0, \infty) \rightarrow \mathbb{R}^{l_{v}}$ is the control; $f: \mathbb{R}^n \rightarrow \mathbb{R}^{n}$ and $g: \mathbb{R}^n \rightarrow \mathbb{R}^{n \times l_v}$ are continuously differentiable on $\mathbb{R}^{n}$.

Let $T >0$ be a prediction horizon, and for each $t \ge0$, define the receding-horizon cost
\begin{align}\label{eq:SJ}
\SJ(t,x,\hat v) &\triangleq 
\int_{t}^{t+T} \frac{1}{2} x(\tau)^\rmT Q(\tau) x(\tau) +\Gamma(\tau)^\rmT x(\tau)\nn\\
&\quad  + \frac{1}{2} \hat v(\tau)^\rmT R_{v}(\tau) \hat v(\tau) +\Omega_{v}(\tau)^\rmT \hat v(\tau) \, \rmd \tau,
\end{align}
where $x$ satisfy \eqref{eq:dyn} with $v=\hat v$, $Q(t) \in \mathbb{R}^{n \times n}$ is positive semidefinite, $R_{v}(t) \in \mathbb{R}^{l_v \times l_v}$ is positive definite, $\Omega_{v}(t) \in \mathbb{R}^{l_v}$, and $\Gamma(t) \in \mathbb{R}^{n}$.

Let $\psi_{0}:\mathbb{R}^{n} \to \mathbb{R}$ be continuously differentiable, and define the \textit{safe set}
\begin{equation}
C_{0}\triangleq \{x \hspace{1mm} | \hspace{1mm} \psi_{0}(x) \geq 0\}, \label{set:C0}
\end{equation}
which is the set of states that satisfy the state constraint. 

Our objective is to construct a feedback control $\hat v$ such that for all time $t \ge 0$, the cost $\SJ(t,x,\hat v)$ is approximately minimized subject to the state constraint $x(t) \in C_0$.
We address the problem using the closed-form C-ADP solution. 
Since \cref{Gral_alg:ctrl_optim,Gral_alg:Lambda_optim,Gral_alg:W_optim,Gral_alg:c_optim,Gral_alg:P_optim,Gral_alg:T_optim,Gral_alg:AK_optim,Gral_alg:lambda_tild_optim,Gral_alg:u_tild_optim,Gral_alg:Fxu_tild_optim} solves the optimization problem \cref{ADP:OCP:Main_prob} that has only one affine constraint at each step of the horizon, we consider a state constraint set that can be modeled as the zero-superlevel set of the single candidate CBF $\psi_{0}$. 
However, this approach can address safe sets that are the intersection of multiple zero-superlevel set (with potentially different relative degrees) by using a barrier function composition method. 
For example, see \cite{safari2025time,rabiee2024composition,rabiee2024closed,rabiee2023automatica,safari2025safe}, which use a log-sum-exponential soft-minimum function to compose multiple barrier functions. 
Similarly, this approach can address both state and input constraints by applying the approach in \cite{rabiee2024composition,rabiee2024closed,kamat2026electromagnetic},
which uses control dynamics to transform input constraints into controller-state constraints.

\subsection{Affine CBF Constraint Function for Safety}

We assume the relative degree of $\psi_{0}$ is at least $d$ at each point in $C_{0}$. 
Specifically, we make the following assumption.
    \begin{assumption}\rm
        There exists a positive integer $d$ such that for all $i \in \{0,1,...,d-2\}$ and all $x \in C_{0}$, $L_{g}L_{f}^{i}\psi_{0}(x)=0$; and there exists $x_{\rm{e}}\in C_{0}$ such that $L_{g}L_{f}^{d-1}\psi_{0}(x_{\rm{e}})\ne0$.
        \label{assum:rel_deg_d}
    \end{assumption}
    
Assumption \ref{assum:rel_deg_d} implies that the relative degree of $\psi_{0}$ is at least $d$ at each point in $C_{0}$. 
We use a higher-order approach to construct a candidate CBF. 
For all $j \in \{1,...,d-1\}$, let $\alpha_{j-1}:\mathbb{R}\to \mathbb{R}$ be a $(d-j)$-times continuously differentiable extended class-$\mathcal{K}$ function, and define
\begin{equation}
        \psi_j(x)\triangleq L_{f}\psi_{j-1}(x)+\alpha_{j-1}(\psi_{j-1}(x)).
        \label{eq:psii}
\end{equation}
For all $j \in \{1,...,d-1\}$, the zero-superlevel set of $\psi_{j}$ is given by
\begin{equation}
        C_{j}\triangleq\{x \hspace{1mm} | \hspace{1mm} \psi_{j}(x) \geq 0\}.
        \label{set:Ci}
\end{equation}
We make the following assumption.
\begin{assumption}\rm 
    For all $x \in \mbox{bd } C_{d-1}$, if $L_{g} \psi_{d-1}(x)=0$, then $L_f \psi_{d-1}(x) >0$. \label{assum:LgLf_neq_0}
\end{assumption}
Assumption~\ref{assum:LgLf_neq_0} is related to the constant-relative-degree assumption that is often invoked for CBF approaches. 
Here, we relax the assumption to require that $L_{g}\psi_{d-1}$ is nonzero only a subset of the boundary of $C_{d-1}$.
Next, define
\begin{equation}
C\triangleq \bigcap_{j=0}^{d-1}C_{j},
        \label{set:C}
\end{equation}
which is assumed to be nonempty and contain no isolated points. 
Assumption \ref{assum:LgLf_neq_0} implies that $C \subseteq C_0$ is control forward invariant.

Next, we construct a CBF-type constraint function on the control to ensure that $x(t) \in C$ for all $t \ge 0$. 
Let $\alpha:\mathbb{R}\to \mathbb{R}$ be locally Lipschitz and nondecreasing such that $\alpha(0)=0$.
Then, consider the CBF state-constraint function $\psi:\mathbb{R}^{n}\times \mathbb{R}^{l_v} \times \mathbb{R} \to \mathbb{R}$ given by
\begin{align}
\psi(x, \hat v, \hat \delta) &\triangleq L_{f}\psi_{d-1}(x)+L_{g}\psi_{d-1}(x) \hat v\nn\\
&\qquad +\alpha(\psi_{d-1}(x))+\psi_{d-1}(x) \hat \delta, \label{eq:DP_const_state}
\end{align}
where $\hat v$ is the control variable and $\hat \delta$ is a slack.
Define
\begin{equation*}
a(x)\triangleq L_{f}\psi_{d-1}(x)+\alpha(\psi_{d-1}(x)), \quad   b(x) \triangleq   \left [ \begin{smallmatrix} L_{g}\psi_{d-1}(x)\\ \psi_{d-1}(x) \end{smallmatrix} \right ],
\end{equation*}
and note that $\psi(x, \hat v, \hat \delta) = a(x) + b(x)^{\rm{T}} \hat u$, where $\hat u = \left [ \begin{smallmatrix} \hat v \\ \hat \delta \end{smallmatrix} \right ]$.
The constraint $a(x) + b(x)^{\rm{T}} \hat u \ge 0$ is sufficient to enforce $x(t) \in C$.

The cost \eqref{eq:SJ} is quadratic in $\hat v$, whereas the constraint $a(x) + b(x)^{\rm{T}} \hat u \ge 0$ is a function of $\hat v$ and $\hat \delta$. 
Thus, let $r_{\delta} > 0$ and consider the regularized cost 
\begin{equation}\label{eq:SJbar}
    \bar \SJ(t,x,\hat u) = \SJ(t,x,\hat v) + \frac{r_{\delta}}{2} \int_t^{t+T} \hat \delta(\tau)^2 \, \rmd \tau.
\end{equation}


\subsection{Receding-Horizon C-ADP Control}

A receding-horizon implementation of \cref{Gral_alg:c_optim,Gral_alg:P_optim,Gral_alg:T_optim,Gral_alg:W_optim,Gral_alg:AK_optim,Gral_alg:ctrl_optim,Gral_alg:Lambda_optim,Gral_alg:u_tild_optim,Gral_alg:Fxu_tild_optim,Gral_alg:lambda_tild_optim} is used to construct a control $\hat v$ and slack $\hat \delta$ that approximately minimizes $\bar \SJ(t,x,\hat u)$ subject to the constraint $\psi(x,\hat v,\hat \delta) \ge 0$. 
To apply \cref{Gral_alg:c_optim,Gral_alg:P_optim,Gral_alg:T_optim,Gral_alg:W_optim,Gral_alg:AK_optim,Gral_alg:ctrl_optim,Gral_alg:Lambda_optim,Gral_alg:u_tild_optim,Gral_alg:Fxu_tild_optim,Gral_alg:lambda_tild_optim}, we discretize \eqref{eq:dyn}.
Let $T_{\rm{p}}>0$ be the planning time step such that $N=T/T_\rmp$ is an integer. 
In other words, $T_\rmp$ is the time step associated with the subscript $i$ in the optimization \cref{ADP:OCP:Main_prob}.
Then, let
    \begin{equation}
        x_{i+1}=f_{\rm{d}}(x_{i})+g_{\rm{d}}(x_{i})v_{i}.
        \label{eq:dyn_sol}
    \end{equation}
be the discretize approximation of \cref{eq:dyn_sol}. 
For example, $f_{\rm{d}}$ and $g_{\rm{d}}$ can be obtained by discretizing \eqref{eq:dyn} using the forward-Euler method. 
In this case, $f_{\rm{d}}(x) = x + T_\rmp f(x)$ and $g_\rmd(x) = T_\rmp g(x)$. 
Higher-order or approximate zero-order-hold methods can also be used to improve accuracy of \eqref{eq:dyn_sol} if needed.
Although $f_{\rm{d}}$ and $g_{\rm{d}}$ are used to construct the optimal functions $u^*_0,\ldots,u^*_N$, the safety constraint $\psi(x,\hat v,\hat \delta) \ge 0$ is based on continuous-time dynamics \eqref{eq:dyn}.  

Since the optimization variable $\hat u$ includes control $\hat v$ and slack $\hat \delta$, let
\begin{equation}\label{eq:F_G}
    F(x) = f_\rmd(x), \qquad G(x) = \begin{bmatrix} g_\rmd(x) & 0_{n\times1} \end{bmatrix}.
\end{equation}
Thus, $x_{i+1} = F(x_i) + G(x_i) u_i$, where $u_i = \left [ \begin{smallmatrix} v_i \\ \delta_i \end{smallmatrix} \right ]$.

The C-ADP functions $u^*_0,\ldots,u^*_N$ are updated using \cref{Gral_alg:c_optim,Gral_alg:P_optim,Gral_alg:T_optim,Gral_alg:W_optim,Gral_alg:AK_optim,Gral_alg:ctrl_optim,Gral_alg:Lambda_optim,Gral_alg:u_tild_optim,Gral_alg:Fxu_tild_optim,Gral_alg:lambda_tild_optim} with update period $T_\rms\in (0,T_\rmp]$.
In other words, $u^*_0,\ldots,u^*_N$ can be updated faster than the sample time used to discretize the horizon.
We use subscript $k\in\BBN$ to indicate the update step. 
For example, $J_{k}$ is the cost for the optimization \cref{ADP:OCP:Main_prob} on step $k$, which corresponds to $kT_\rmp$ in real time. 
Similarly, $u^*_{0,k},\ldots,u^*_{N,k}$ are the optimal functions obtained from \cref{Gral_alg:c_optim,Gral_alg:P_optim,Gral_alg:T_optim,Gral_alg:W_optim,Gral_alg:AK_optim,Gral_alg:ctrl_optim,Gral_alg:Lambda_optim,Gral_alg:u_tild_optim,Gral_alg:Fxu_tild_optim,Gral_alg:lambda_tild_optim} on time step $k$.
For each $k\in\BBR$, $J_{k}$ is obtained by discretizing \eqref{eq:SJbar}. 
Specifically, $J_{k}$ is given by \cref{ADP:OCP:Main_prob}, where
\begin{gather*}
    Q_{i,k} = Q(kT_\rms+iT_\rmp), \quad     \Gamma_{i,k} = \Gamma(kT_\rms+iT_\rmp),\\
    R_{i,k} = \begin{bmatrix} R_{v}(kT_\rms+iT_\rmp) & 0_{l_v \times 1} \\ 0_{1 \times l_v}& r_{\delta} \end{bmatrix}, \,\,     \Omega_{i,k} = \begin{bmatrix} \Omega(kT_\rms+iT_\rmp) \\ 0 \end{bmatrix}.
\end{gather*}
Then, for each $k\in\BBN$, the optimal functions $u^*_{0,k},\ldots,u^*_{N,k}$ are obtained from \cref{Gral_alg:c_optim,Gral_alg:P_optim,Gral_alg:T_optim,Gral_alg:W_optim,Gral_alg:AK_optim,Gral_alg:ctrl_optim,Gral_alg:Lambda_optim,Gral_alg:u_tild_optim,Gral_alg:Fxu_tild_optim,Gral_alg:lambda_tild_optim}, where $a_i=a$, $b_i=b$, and $\{\bar{x}_{i,k}\}_{i\in \mathcal{N}}$ is the nominal trajectory used on the $k$th step. 

Finally, for all $k \in \mathbb{N}$ and all $t \in [kT_{\rm{s}},(k+1)T_{\rm{s}})$, the optimal control $v_*:[0,\infty) \times \mathbb{R}^{n}\rightarrow \mathbb{R}^{\ell}$ is given by
\begin{equation}\label{eq:v*}
    v_*(t,x) \triangleq \begin{bmatrix}  I_{l_v} & 0_{l_v} \end{bmatrix} u^*_{0,k}(x),
\end{equation}
and the optimal slack parameter  $\delta_*:[0,\infty) \times \mathbb{R}^{n}\rightarrow \mathbb{R}^{\ell}$ is 
\begin{equation}
    \delta_*(t,x) \triangleq \begin{bmatrix}  0_{1 \times l_v} & 1 \end{bmatrix} u^*_{0,k}(x),
\end{equation}
Note that $v_*$ and $\delta_*$ are piecewise continuous in $t$ and continuous in $x$ on $\BBR^n$. 
In addition, $v_*$ and $\delta_*$ are continuous in $t$ on each interval $[kT_\rms,(k+1)T_\rm)$ but are generally discontinuous at $kT_\rms$, which corresponds to the $u^*_{0,k}$ update. 
A smooth homotopy between consecutive optimal functions  (e.g., 
\cite{safari2025time,rabiee2023automatica}) can be used instead of the the update \eqref{eq:v*}.
If $F^\prime$, $G^\prime$, and $\psi_{d-1}^\prime$ are locally Lipschitz on $\SD \subset \BBR^n$, then \cref{Gral_alg:c_optim,Gral_alg:W_optim,Gral_alg:ctrl_optim,Gral_alg:Lambda_optim} imply that $v_*$ and $\delta_*$ are locally Lipschitz on $\SD$.

The following theorem is the main result constraint satisfaction.
The proof is similar to \cite{safari2025time} and is omitted due to space limitations.

\begin{theorem} \label{thm:safety} \rm
Consider \eqref{eq:dyn}, where Assumptions~\ref{assum:rel_deg_d} and \ref{assum:LgLf_neq_0} are satisfied.
Let $v=v_*$, where $v_*$ is given by \eqref{eq:v*}.
Assume $F^\prime$, $G^\prime$, and $\psi_{d-1}^\prime$ are locally Lipschitz on $C$.
Then, for all $x_0 \in C$, the following hold:
\begin{enumerate}

\item There exists $t_{\rm{m}}>0$ such that \eqref{eq:dyn} with $v=v_*$ has a unique solution on $[0,t_{\rm{m}})$.

\item For all $t\in[0,t_{\rm{m}})$, $x(t)\in C \subseteq C_{0}$.

\end{enumerate}
\end{theorem}

Theorem~\ref{thm:safety} shows that $v_*$ ensures safety, and Theorem~\ref{thm:optimal} describes the optimality property of the sequence of functions $u^*_{0,k}$ that constitute $v_*$.

In this receding-horizon implementation, the nominal trajectory $\{\bar{x}_{i,k}\}_{i\in \mathcal{N}}$ is obtained on each update $k$ in a forward pass through \eqref{eq:dyn_sol} using the latest optimal control functions $u^{*}_{i,k-1}$.
Specifically, let $\{\bar{x}_{i,0}\}_{i\in \mathcal{N}}$ be the initial nominal trajectory, and let $\{u^{*}_{i,0}\}_{i \in \mathcal{N}}$ be the solution to \cref{Gral_alg:c_optim,Gral_alg:P_optim,Gral_alg:T_optim,Gral_alg:W_optim,Gral_alg:AK_optim,Gral_alg:ctrl_optim,Gral_alg:Lambda_optim,Gral_alg:u_tild_optim,Gral_alg:Fxu_tild_optim,Gral_alg:lambda_tild_optim} using $\{\bar{x}_{i,0}\}_{i\in \mathcal{N}}$.
Then, for all $k \in \{1,2,...\}$, $\{\bar{x}_{i,k}\}_{i\in \mathcal{N}}$ is the solution to
\begin{equation}
        \bar{x}_{i+1,k}=F(\bar{x}_{i,k})+G(\bar{x}_{i,k})u^{*}_{i,k-1}(\bar{x}_{i,k}), \quad \bar{x}_{0,k}=x(kT_{\rm{s}}). 
        \label{eq:cbf_ctrl:forward_pass}
\end{equation}
Receding-horizon C-APD control is summarized in Alg.~\ref{alg:CADP}.


    \begin{algorithm}[t]
    \DontPrintSemicolon
    \caption{Receding-Horizon C-ADP Control}
    \label{alg:CADP}
    \KwIn{$\{\bar{x}_{i,0}\}_{i\in\mathcal{N}},Q,\Gamma,R_{v},\Omega_{v},r_{\delta},T_{\rm p},T,a,b,F, G$}
    \KwOut{$v_{*}(t,x)$}
    \BlankLine
    \For{$k \leftarrow 0,\,1,\,2,\,\ldots$}{
        \If{$k > 0$}{
            \tcp{Forward pass}
            $\{\bar{x}_{i,k}\}_{i \in \mathcal{N}}
                \;\leftarrow\; $
                \textit{\cref{eq:cbf_ctrl:forward_pass}}\;
        }
        \tcp{Backward pass}
        $\{u^{*}_{i,k}\}_{i \in \mathcal{N}}
            \;\leftarrow\;$            \textit{\cref{Gral_alg:ctrl_optim,Gral_alg:Lambda_optim,Gral_alg:c_optim,Gral_alg:W_optim,Gral_alg:c_optim,Gral_alg:P_optim,Gral_alg:T_optim,Gral_alg:AK_optim,Gral_alg:u_tild_optim,Gral_alg:Fxu_tild_optim,Gral_alg:lambda_tild_optim}}\;

        \tcp{Extract control function}
        $v_{*}(t,x)\;\leftarrow\; $ \textit{\cref{eq:v*}}
    }
    \end{algorithm}

\section{Nonholonomic Ground Robot}
\label{sec:robot}

    Consider the nonholonomic differential drive mobile robot modeled by \eqref{eq:dyn}, where
    \begin{gather*}
    x = \begin{bmatrix}
                q_{\rm{x}}\\
                q_{\rm{y}}\\
                \gamma\\
                s \\
                \omega
    \end{bmatrix}, \qquad
    f(x)=\begin{bmatrix}
                s\cos{\gamma}-l_{\rm{d}}\omega\sin{\gamma} \\
                s\sin{\gamma}+l_{\rm{d}}\omega\cos{\gamma} \\
                \omega \\
                -c_1 s - c_2 s^2 \tanh \frac{s}{\varepsilon_1}  \\
                -c_3 \omega - c_4 \omega^2 \tanh \frac{\omega}{\varepsilon_2}
    \end{bmatrix}, \\
    g(x) = \begin{bmatrix} 0_{3 \times 2} \\ M \end{bmatrix},\quad 
        M \triangleq \frac{k_{\rmm}}{rR_\rma} \begin{bmatrix} 
            \frac{1}{m} & \frac{1}{m} \\
             \frac{l}{I} &  - \frac{l}{I}
        \end{bmatrix}, \quad
    u = \begin{bmatrix}
    u_{\rm{r}} \\ u_{\rm{l}}
        \end{bmatrix},\\
       c_1\triangleq\frac{2k_{\rm{b}}k_{\rm{m}}}{mrR_{\rm{a}}}+\frac{2\varepsilon_3}{mr},\quad   
     c_3\triangleq\frac{k_{\rm{b}}k_{\rm{m}}l^2}{Ir^2R_{\rm{a}}}+\frac{l\varepsilon_3}{Ir^2},  
    \end{gather*}
    where 
    $q\triangleq [q_{\rm{x}} \quad q_{\rm{y}}]^{\rm{T}}$ is the position of a point of interest on the robot in an orthogonal coordinate frame, 
    $\gamma$ is the direction of the velocity vector, 
    $s$ is the speed,
    $\omega$ is the angular velocity, 
    $u_{\rm{r}}$ and $u_{\rm{l}}$ are the control voltage applied to the right and left motors, 
    $k_{\rm{m}}=0.1$ N-m/Amp is the torque constant, 
    $r=0.1$~m is the wheel radius, 
    $l=0.5$~m is the distance between wheels, 
    $l_{\rm{d}}=0.25$~m is the distance from the mass center to the point of interest, 
    $R_{\rm{a}}=0.27$~ohms is the armature resistance, 
    $m=10$~kg is the mass,
    $I=0.83$ kg-$\rm{m^{2}}$ is the moment of inertia, 
    $k_{\rm{b}} = 0.0487$~V-s/rad is the motor back-EMF constant,
    $\varepsilon_3 = 0.01$~N-m-s is the friction coefficient,
    $c_2=0.4581$ $\mathrm{m}^{-1}$,
    $c_4 = 0.3477$, and
    $\varepsilon_{1}=\varepsilon_{2} =0.4$. 
    This model includes linear damping from the back-EMF of the motor and friction of the ground, as well as drag damping. 
    This model is from \cite{anvari2013non}.

Consider the map in Fig.~\ref{fig:GR:map}.
The objective is to navigate to a goal state $x_\rmd \triangleq [ \, q_\rmd^\rmT \quad 0_{3\times 1} \, ]^\rmT$, where $q_{\rm{d}} \in \BBR^3$ is the goal loation, and satisfy state constraints. 
The robot must not collide with obstacles, its speed must satisfy $|s| \leq \bar{s}=1.5$ m/s, and its angular velocity must satisfy $|\omega|\leq \bar{\omega}=0.5$ rad/s. 
These constraints are modeled with barrier functions.
Specifically, the area outside each obstacle is modeled as the zero-superlevel set of $\phi_{1}(x),...,\phi_{41}(x)$, and the bounds on $s$ and $\omega$ are modeled by $h_{42}(x)\triangleq \bar{s}^2-s^2$ and $h_{43}(x)\triangleq\bar{\omega}^2-\omega^2$.
Note that $h_{42}$ and $h_{43}$ have relative degree 1, whereas $\phi_{1},...,\phi_{41}$ have relative degree 2.
To address multiple barrier functions with different relative degrees, we adopt the approach in \cite{rabiee2024closed,rabiee2024composition}, which uses a log-sum-exponential function to compose multiple barrier functions. 
To compose the barrier functions with different degrees, we use a higher-order approach with $\phi_{1},...,\phi_{41}$, because they have relative degree two. 
For all $i\in \{1,2,...,41\}$, define $h_{i}(x)\triangleq L_{f}\phi_{i}(x)+\zeta\phi_{i}(x)$, where $\zeta=0.5$, and note that $h_i$ is relative degree one. 
Then, let $\rho>0$, and consider the safe set $C_{0}$ given by \eqref{set:C0}, where
\begin{equation}
\psi_{0}(x)\triangleq - \frac{1}{\rho}\log\Big(\sum_{i=1}^{n_h} \textrm{e}^{-\rho h_{i}(x)}\Big),
        \label{eq:GR:CBF}
\end{equation}
where $n_h=43.$
The log-sum-exponential \eqref{eq:GR:CBF} is a smooth lower-bound approximation of $\min\{h_{1},...,h_{n_h}\}$ with worst-case approximation error bounded by $\rho^{-1}\log n_h$ \cite{safari2025time}.
Selecting $\rho$ is a trade-off between the conservativeness of the approximation and numerical conditioning (see \cite{safari2025time,rabiee2024composition,rabiee2024closed}). 
We let $\rho=750$.


We compare receding-horizon C-ADP control with other several approaches.
Specifically, we compare the following:

\noindent \textbf{Method 1. C-ADP}:
Algorithm~\ref{alg:CADP} is implemented with $\eta=1$, $Q = \mathrm{diag}(1, 1, 0, 16,160)$, $\Gamma = Q x_{\rm{d}}$, $R_v=\mathrm{diag}(80,80)$, $\Omega_v = 0_{2 \times 1}$, and $r_{\delta}= 0.2\times10^{10}$.
The horizon is $T=20$~s with $T_\rmp=0.05$~s, which yields $N=400$, and $T_\rms = T_\rmp$. 


\noindent \textbf{Method 2. CiLQR+CBF}: 
CiLQR from \cite{iLQR_Bharath2} is implemented to compute a desired control $v_{\mathrm{d}}$ using a single forward-and-backward pass per update, Levenberg–Marquardt regularization, and the line-search capability in \cite{iLQR_Bharath2} disabled. 
The horizon $N$, planning time step $T_\rmp$, update period $T_\rms$, and cost are the same as in Method~1. 
As common in CiLQR, the constraints are added to the cost as exponentials (i.e., $e^{-h_{42}(x)}+e^{-h_{43}(x)}+\sum_{i=1}^{41} e^{-\phi_i(x)}$). 
Since CiLQR does not guarantee constraint satisfaction, the implemented control $v$ is computed from a minimum-intervention QP with the CBF constraint \eqref{eq:DP_const_state}.
Specifically, 
\begin{equation}\label{GR:alg:CilQR_1}
\begin{aligned}
& (v(x),\delta(x))=\arg \min_{\hat v,\hat \delta} \| \hat v - v_{\rm{d}}(x)\|^{2} +r_{\delta} \hat \delta^2 \\
            & \qquad \mbox{subject to \eqref{eq:DP_const_state}}.
\end{aligned}
\end{equation}


\noindent \textbf{Method 3. CiLQR+LS+CBF}:
This is the same as Method~2, but the  line-search capability is enabled, where up to 10 forward-and-backward passes are allowed per update.

\noindent \textbf{Method 4. NaiveControl+CBF}:
This method uses an analytic desired control $v_\rmd$ that accomplishes destination seeking in an obstacle-free environment. 
However, $v_\rmd$ is naive to obstacles, so the implemented control $v$ is computed from the minimum intervention QP \eqref{GR:alg:CilQR_1}. 
This desired control is 
\begin{equation*}
        \resizebox{0.9\columnwidth}{!}{$
        v_{\rm{d}}(x) \triangleq \big(M^{\rm{T}}M\big)^{-1}M^{\rm T}
        \begin{bmatrix}
            \dot{s}_{\rm{d}}(x)-c_1 s - c_2 s^2 \tanh \frac{s}{\varepsilon_1} \\
            \dot{\omega}_{\rm{d}}(x)-c_3 \omega - c_4 \omega^2 \tanh \frac{\omega}{\varepsilon_2}
        \end{bmatrix},
        $}
\end{equation*}
where $k_{\rm{p}}=0.4$, $k_{\rm{d}}=0.75$, and 
\begin{align*}
\dot{s}_{\rm{d}}(x) &\triangleq
                    \begin{bmatrix}
                        \cos{\gamma} && \sin{\gamma}
                    \end{bmatrix} a_{\rm{d}} + l_{\rm{d}}\omega^{2},\\
\dot{\omega}_{\rm{d}}(x) &\triangleq
                   l_{\rm{d}}^{-1}\left ( \begin{bmatrix}
                            -\sin{\gamma} &&  \cos{\gamma}
                        \end{bmatrix} a_{\rm{d}}-s\omega \right ),\\
a_{\rm{d}}(x)&\triangleq -k_{\rm{p}}\tanh \left ( q - q_{\rm{d}} \right ) - k_{\rm{d}} \tanh 
            \dot q.
\end{align*}

C-ADP control naturally satisfies the constraints (see Theorem~\ref{thm:safety}).
For the other methods, constraints are not necessarily satisfied by the desired control $v_\rmd$, which is why the CBF filter \eqref{GR:alg:CilQR_1} is used.
For all methods, the control $v$ is implemented with a zero-order hold at 100 Hz.

   \begin{figure}[t]
        \centering
        \includegraphics[width=\columnwidth]{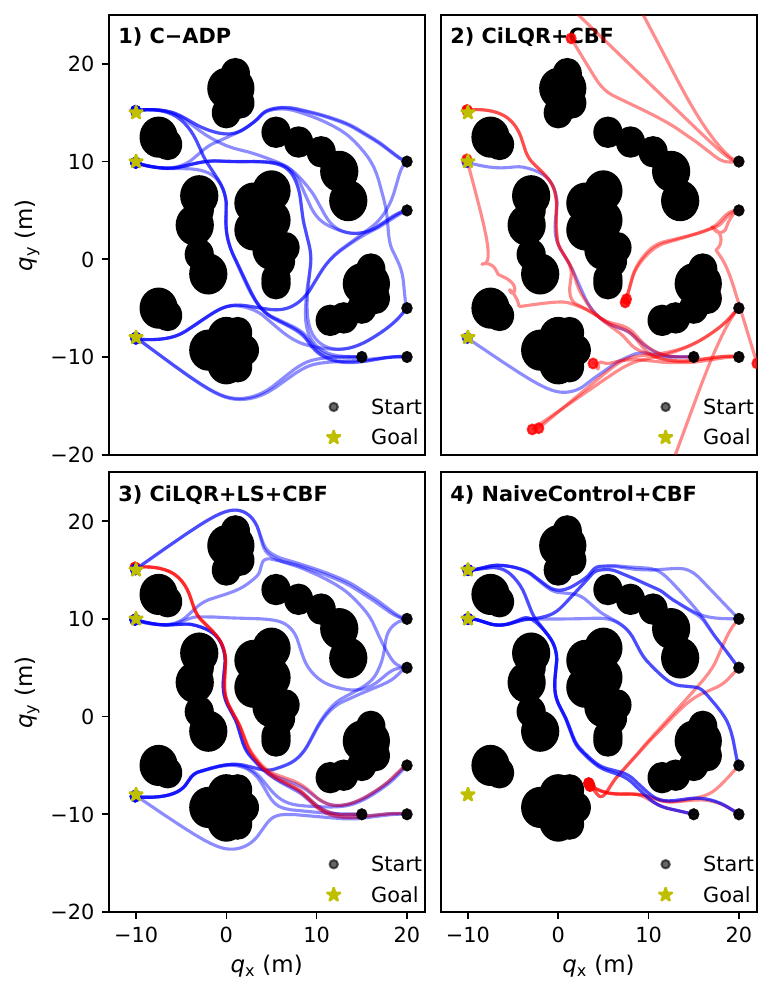}
        \caption{All 15 trajectories for each method. 
        Blue lines indicate trajectories that reach $q_\rmd$, whereas red lines do not.}
        \label{fig:GR:map}
    \end{figure}

Each method is evaluated using 15 trials (5 initial positions paired with 3 goal locations).
We say that the robot reaches $q_\rmd$ if the distance error $\| q - q_\rmd \|$ reaches and stays within $d_{\rm tol} = 0.25$~m by the final time $T_{\rm{f}}=120$~s. 
Figure~\ref{fig:GR:map} shows the map with the 15 trajectories for each method. 
For each method, the robot safely navigate obstacles.
All trials with C-ADP reach $q_\rmd$, whereas the success rates with Methods~2, 3, and 4 are 2, 12, and 10 out of 15, respectively. 
For Method 3, the 3 trials do not reach $q_\rmd$ reach $0.5$~m error; however, they do not get closer.



The computation time (mean, standard deviation) on an Acer Nitro AN515-57 laptop (Intel Core i5-11400H @ 2.70 GHz, 16 GB RAM) are $(51, 2.1)$ ms, $(8.1, 0.5)$ ms, $(62, 19)$ ms, and $(2.2, 0.2)$ ms for Methods~1, 2, 3 and 4, respectively. 
Thus, C-ADP compares favorably with Method 3 (CiLQR with line search), which is the only method with comparable performance. 
Note that CiLQR from \cite{iLQR_Bharath2} is computationally optimized, whereas our implementation of C-ADP is not. 
The large standard deviation for Method 3 is due to adaptive regularization and line search, which vary the number of forward-backward passes per step.

To compare the methods in more detail, we consider 3 performance metrics. 
First, the success indicator (SI), which is equal to 0 if the robot reaches $q_\rmd$ and equal to 1 otherwise. 
Second, the arrival time (AT), which is the time at which the robot reaches $q_\rmd$.
We let $\textrm{AT}=T_{\rm{f}}$ if the robot does not reach $q_{\rm{d}}$. 
Third, the total cost is 
\begin{equation*}
\textrm{TC} \triangleq \int_0^{T_\rmf} [x(t)-x_\rmd]^\rmT R [x(t)-x_\rmd] + v(t)^\rmT Q_v v(t) \, \rmd t.
\end{equation*}
For each method and each trial, we also consider 3 metrics related to the control.
The control magnitude (CM), control derivative (CD), and control intervention (CI) are 
    \begin{equation*}
        \textrm{CM} \triangleq \max_{t} \| v(t) \| , \,
        \textrm{CD} \triangleq \max_{t} \| \dot{v}(t) \|, \,
        \textrm{CI} \triangleq \max_{t} \| \delta v(t) \|,
    \end{equation*}
where $\delta v \triangleq v-v_{\mathrm{d}}$. 
For C-ADP control, it follows from \cref{Gral_alg:ctrl_optim,Gral_alg:Lambda_optim,Gral_alg:W_optim,Gral_alg:c_optim,Gral_alg:P_optim,Gral_alg:T_optim} that $v_\rmd$ is the value obtained with $\lambda_{0,k}=0$. 
Next, we normalize these metrics across all 60 trials (4 methods, 15 trials). 
For each $j \in I \triangleq \{1,2,\ldots,60\}$, let 
\begin{equation*}
\tilde{\rm AT}_{j} \triangleq  
        \frac{ \max_{\ell \in I} \, {\rm AT}_{\ell}-{\rm AT}_{j}}{ \max_{\ell \in I} \, {\rm AT}_{\ell}-\min_{\ell \in I} \, {\rm AT}_{\ell}}.
\end{equation*}
Note that $\tilde{\rm AT}_{j} =1$ if trial $j$ arrives the fastest of the 60 trials, whereas $\tilde{\rm AT}_{j}=0$ if trial $j$ arrives the slowest (i.e., does not arrive). 
Similarly, let $\tilde{\rm FI}_{j}$, $\tilde{\rm TC}_{j}$, $\tilde{\rm CM}_{j}$, $\tilde{\rm CD}_{j}$, and $\tilde{\rm CI}_{j}$ denote the other normalized metrics. 
For each metric, 1 is best and 0 is worst. 
Figure~\ref{fig:GR:perf_rad} provides radar plots summarizing the metrics for each algorithm. 
The shaded region is the average across all 15 trials and the dashed lines show each trial.

Figure~\ref{fig:GR:perf_rad} shows that C-ADP provides the best overall results followed by CiLQR with line search. 
In contrast, CiLQR without line search has the worst performance metrics (TC, FI, and AT). 
It is notable that line search significantly improves the performance of CiLQR because line-search approaches can also be implemented with the C-ADP control method. 
We note that since the C-ADP control has a closed-form expression, its can be implemented with a faster zero-order hold without reducing the update time $T_\rms \in (0,T_\rmp]$. 
In contrast, CiLQR is limited by $T_\rms$ in the sense that the control obtained from CiLQR is constant for $T_\rms$ (although the minimum intervention CBF can be implemented faster as it was in this example). 


 \begin{figure}
       \centering
       \includegraphics[width=1.0\linewidth]{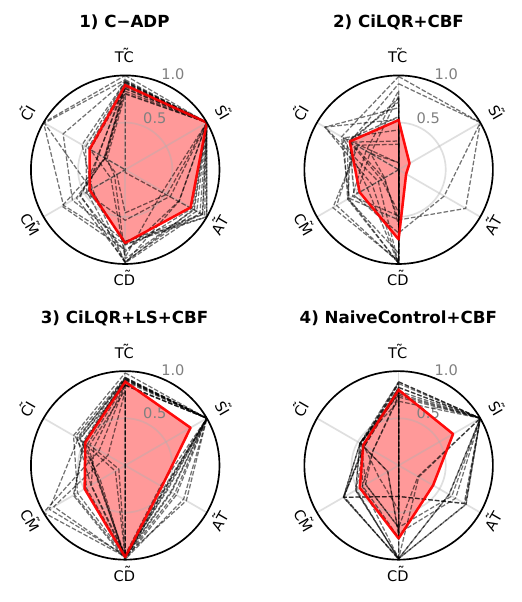}
       \caption{Performance and control metrics.
       Red lines show the mean, and the dashed lines show the 15 trials.}
       \label{fig:GR:perf_rad}
\end{figure}

    \bibliographystyle{ieeetr}
    \bibliography{References}

@article{bertsekas2011dynamic,
  title={Dynamic programming and optimal control 3rd edition, volume ii},
  author={Bertsekas, Dimitri P and others},
  journal={Belmont, MA: Athena Scientific},
  volume={1},
  year={2011}
}

@inproceedings{ames2019control,
  title={Control barrier functions: Theory and applications},
  author={Ames, Aaron D and Coogan, Samuel and Egerstedt, Magnus and Notomista, Gennaro and Sreenath, Koushil and Tabuada, Paulo},
  booktitle={Proc. Euro. Contr. Conf.},
  pages={3420--3431},
  year={2019}
}

@inproceedings{cosner2021measurement,
  title={Measurement-robust control barrier functions: Certainty in safety with uncertainty in state},
  author={Cosner, Ryan K and Singletary, Andrew W and Taylor, Andrew J and Molnar, Tamas G and Bouman, Katherine L and Ames, Aaron D},
  booktitle={IEEE/RSJ Int. Conf. Int. Robots Sys.},
  pages={6286--6291},
  year={2021}
}

@article{breeden2023robust,
  title={Robust control barrier functions under high relative degree and input constraints for satellite trajectories},
  author={Breeden, Joseph and Panagou, Dimitra},
  journal={Automatica},
  volume={155},
  pages={111109},
  year={2023},
  publisher={Elsevier}
}

@article{xiao2022event,
  title={Event-triggered control for safety-critical systems with unknown dynamics},
  author={Xiao, Wei and Belta, Calin and Cassandras, Christos G},
  journal={IEEE Trans. Autom. Contr.},
  volume={68},
  number={7},
  pages={4143--4158},
  year={2022},
  publisher={IEEE}
}

@article{wang2023multi,
  title={Multi-UAV Safe Collaborative Transportation Based on Adaptive Control Barrier Function},
  author={Wang, Zhijun and Hu, Tengfei and Long, Lijun},
  journal={IEEE Trans. Sys., Man, and Cyb.: Systems},
  year={2023},
  publisher={IEEE}
}

@article{zheng2023constrained,
  title={Constrained moving path following control for UAV with robust control barrier function},
  author={Zheng, Zewei and Li, Jiazhe and Guan, Zhiyuan and Zuo, Zongyu},
  journal={IEEE/CAA Automatica Sinica},
  volume={10},
  number={7},
  pages={1557--1570},
  year={2023},
  publisher={IEEE}
}

@article{rabiee2023automatica,
  title={Soft-minimum and soft-maximum barrier functions for safety with actuation constraints},
  author={Rabiee, Pedram and Hoagg, Jesse B},
  journal={Automatica},
  year={2024}
}

@techreport{anvari2013non,
  title={Non-holonomic differential drive mobile robot control \& design: Critical dynamics and coupling constraints},
  author={Anvari, Iman},
  year={2013},
  institution={Arizona State University}
}

@article{rabiee2024closed,
  title={A Closed-Form Control for Safety Under Input Constraints Using a Composition of Control Barrier Functions},
  author={Rabiee, Pedram and Hoagg, Jesse B},
  journal={arXiv preprint arXiv:2406.16874},
  year={2024}
}

@article{safari2025time,
  title={Time-varying soft-maximum barrier functions for safety in unmapped and dynamic environments},
  author={Safari, Amirsaeid and Hoagg, Jesse B},
  journal={IEEE Trans. on Contr. Sys. Tech.},
  year={2025},
  publisher={IEEE}
}

@article{ames2016control,
  title={Control barrier function based quadratic programs for safety critical systems},
  author={Ames, Aaron D and Xu, Xiangru and Grizzle, Jessy W and Tabuada, Paulo},
  journal={IEEE Trans. Autom. Contr.},
  volume={62},
  number={8},
  pages={3861--3876},
  year={2016},
  publisher={IEEE}
}

@article{wieland2007constructive,
  title={Constructive safety using control barrier functions},
  author={Wieland, Peter and Allg{\"o}wer, Frank},
  journal={IFAC Proc. Volumes},
  volume={40},
  number={12},
  pages={462--467},
  year={2007},
  publisher={Elsevier}
}

@book{borrelli2017predictive,
  title={Predictive control for linear and hybrid systems},
  author={Borrelli, Francesco and Bemporad, Alberto and Morari, Manfred},
  year={2017},
  publisher={Cambridge University Press}
}

@article{bemporad2002model,
  title={Model predictive control based on linear programming\~{} the explicit solution},
  author={Bemporad, Alberto and Borrelli, Francesco and Morari, Manfred and others},
  journal={IEEE Trans. Autom. Contr.},
  volume={47},
  number={12},
  pages={1974--1985},
  year={2002}
}

@article{tondel2003algorithm,
  title={An algorithm for multi-parametric quadratic programming and explicit {MPC} solutions},
  author={T{\o}ndel, Petter and Johansen, Tor Arne and Bemporad, Alberto},
  journal={Automatica},
  volume={39},
  number={3},
  pages={489--497},
  year={2003},
  publisher={Elsevier}
}

@article{prajna2007framework,
  title={A framework for worst-case and stochastic safety verification using barrier certificates},
  author={Prajna, Stephen and Jadbabaie, Ali and Pappas, George J},
  journal={IEEE Trans. Autom. Contr.},
  volume={52},
  number={8},
  pages={1415--1428},
  year={2007},
  publisher={IEEE}
}

@inproceedings{sleiman2021constraint,
  title={Constraint handling in continuous-time ddp-based model predictive control},
  author={Sleiman, Jean-Pierre and Farshidian, Farbod and Hutter, Marco},
  booktitle={Int. Conf. on Rob. and Autom. (ICRA)},
  pages={8209--8215},
  year={2021},
  organization={IEEE}
}

@inproceedings{aoyama2021constrained,
  title={Constrained differential dynamic programming revisited},
  author={Aoyama, Yuichiro and Boutselis, George and Patel, Akash and Theodorou, Evangelos A},
  booktitle={Int. Conf. on Rob. and Autom. (ICRA)},
  pages={9738--9744},
  year={2021},
  organization={IEEE}
}

@article{mare2007solution,
  title={Solution of the input-constrained LQR problem using dynamic programming},
  author={Mare, Jos{\'e} B and De Don{\'a}, Jos{\'e} A},
  journal={Sys, Contr. Letts.},
  volume={56},
  number={5},
  pages={342--348},
  year={2007},
  publisher={Elsevier}
}

@book{rao2019engineering,
  title={Engineering optimization: {T}heory and practice},
  author={Rao, Singiresu S},
  year={2019},
  publisher={Wiley}
}

@article{ha2022novel,
  title={A novel value iteration scheme with adjustable convergence rate},
  author={Ha, Mingming and Wang, Ding and Liu, Derong},
  journal={IEEE Trans. on Neur. Net. and Learn. Sys.},
  volume={34},
  number={10},
  pages={7430--7442},
  year={2022},
  publisher={IEEE}
}

@article{xiao2015online,
  title={Online optimal control of unknown discrete-time nonlinear systems by using time-based adaptive dynamic programming},
  author={Xiao, Geyang and Zhang, Huaguang and Luo, Yanhong},
  journal={Neurocomputing},
  volume={165},
  pages={163--170},
  year={2015},
  publisher={Elsevier}
}

@article{wang2010adaptive,
  title={Adaptive Dynamic Programming for Finite-Horizon Optimal Control of Discrete-Time Nonlinear Systems With $\varepsilon$-Error Bound},
  author={Wang, Fei-Yue and Jin, Ning and Liu, Derong and Wei, Qinglai},
  journal={IEEE Trans. on Neur. Net.},
  volume={22},
  number={1},
  pages={24--36},
  year={2010}
}

@article{guo2017policy,
  title={Policy approximation in policy iteration approximate dynamic programming for discrete-time nonlinear systems},
  author={Guo, Wentao and Si, Jennie and Liu, Feng and Mei, Shengwei},
  journal={IEEE Trans. on Neur. Net. and Learn. Sys.},
  volume={29},
  number={7},
  pages={2794--2807},
  year={2017}
}

@inproceedings{rabiee2025guaranteed,
  title={Guaranteed-safe mppi through composite control barrier functions for efficient sampling in multi-constrained robotic systems},
  author={Rabiee, Pedram and Hoagg, Jesse B},
  booktitle={Proc. Conf. Dec. Contr.},
  pages={5515--5520},
  year={2025}
  }

@misc{iLQR_Bharath2,
  author       = {Bharath Irigireddy},
  title        = {i{LQR}: Iterative Linear Quadratic Regulator implementation},
  howpublished = {\url{https://github.com/Bharath2/iLQR}},
  year         = {2020},
  note         = {Accessed: 2026-03-19},
}

@inproceedings{rabiee2024composition,
  title={Composition of control barrier functions with differing relative degrees for safety under input constraints},
  author={Rabiee, Pedram and Hoagg, Jesse B},
  booktitle={Proc. Amer. Contr. Conf.},
  pages={3692--3697},
  year={2024},
}

@article{yuan2020solver,
  title={Solver--critic: A reinforcement learning method for discrete-time-constrained-input systems},
  author={Yuan, Xin and Dong, Lu and Sun, Changyin},
  journal={IEEE Trans. Cybern.},
  volume={51},
  number={11},
  pages={5619--5630},
  year={2020}
}

@article{yang2019approximate,
  title={Approximate dynamic programming for nonlinear-constrained optimizations},
  author={Yang, Xiong and He, Haibo and Zhong, Xiangnan},
  journal={IEEE Trans. Cybern.},
  volume={51},
  number={5},
  pages={2419--2432},
  year={2019}
}

@article{kamat2026electromagnetic,
  title={Electromagnetic Formation Flying Using Alternating Magnetic Field Forces and Control Barrier Functions for State and Input Constraints},
  author={Kamat, Sumit S and Seigler, T Michael and Hoagg, Jesse B},
  journal={IEEE Trans. Aerosp. Electron. Syst.},
  year={2026}
}

@inproceedings{safari2025safe,
  title={Safe Navigation in Unmapped Environments for Robotic Systems with Input Constraints},
  author={Safari, Amirsaeid and Hoagg, Jesse B},
  booktitle={Proc. Conf. Dec. Contr.},
  pages={6957--6962},
  year={2025},
  organization={IEEE}
}

@inproceedings{li2004iterative,
  title={Iterative linear quadratic regulator design for nonlinear biological movement systems},
  author={Li, Weiwei and Todorov, Emanuel},
    booktitle={Int. Conf. Informatics Control Autom. Robot.},
  volume={2},
  pages={222--229},
  year={2004}
}

 \end{document}